\documentclass[cis]{ipart_v1}

\Vol{17}
\Issue{4}
\Year{2017}
\firstpage{219}

\usepackage[english]{babel}

\RequirePackage[OT1]{fontenc}
\RequirePackage{amsthm,amsmath,amssymb}
\RequirePackage{xcolor}

\newtheorem{thm}{Theorem}
\newtheorem{lemma}{Lemma}

\begin{document}

\title[Part II: stationary sources satisfying $\psi$-mixing criterion]{Layered black-box, behavioral interconnection perspective and applications to problems in communications, Part II: stationary sources satisfying $\psi$-mixing criterion}

\author[M. Agarwal, S. Mitter, and A. Sahai]{Mukul Agarwal, Sanjoy Mitter, and Anant Sahai}

\begin{abstract}
Theorems from Part 1 of this paper are generalized to stationary, $\psi$-mixing sources in this paper. As a consequence, these theorems are proved for Markoff chains and order $m$ Markoff chains. The main result is the generalization of Theorem 1 in Part 1.
\end{abstract}

\maketitle

\section{Introduction}

In this paper, we generalize results Theorems 1-4 from Part 1 of this paper \cite{paperTogether} to the case when the source $X$ is not necessarily i.i.d. but stationary and satisfies a mixing condition, the $\psi$-mixing criterion (which implies that the process is also ergodic). As a corollary, the results hold for Markoff chains and order $m$ Markoff chains.

In Part 1 of this paper, a direct equivalence was drawn between the problem of communicating an i.i.d. source $X$ to within a certain distortion level $D$ over an essentially unknown channel and reliable communication at rates less than the rate-distortion function $R_X(D)$ over the channel. As a result, assuming random codes are permitted, a source-channel separation theorem was proved for communication over a general, compound channel, where the channel model is general in the sense of Verdu-Han and compound in the sense that the channel may belong to a set. These theorems were then generalized to the unicast, multi-user setting where the sources were still assumed to be i.i.d.

In this paper, these theorems from Part 1 are generalized to the case when the source (also sources in the unicast multi-user setting) are not necessarily i.i.d. but satisfy a mixing criterion called the $\psi$-mixing criterion.

\section{Paper outline}

In Section \ref{S3}, the notation and deﬁnitions used in this paper are described. This is followed by a description of the $\psi$-mixing condition, its properties and intuition for it in Section~\ref{TheMix}; the proofs of these properties can be found in Appendix \ref{Append}. In Section~\ref{S5}, the high-level idea of the proof of the generalization of Theorem 1 of \cite{paperTogether} to this paper is stated. A simulation procedure is required in order to bring this high-level idea to fruition and this is the subject of Section~\ref{S6}. This is followed by the statement of the main lemma of this paper, Lemma \ref{MainMainLemma}, which uses the simulation procedure of the previous section to prove a result which is the heart of this paper and the heart of what is needed in order to generalize Theorem 1 of \cite{paperTogether} to $\psi$-mixing sources: this is the subject of Section~\ref{MainLemma}. Lemma \ref{MainMainLemma} and a technical lemma relating rate-distortion functions under the expected and the probability of excess distortion criteria is needed in order to generalize Theorem 1 of \cite{paperTogether} to $\psi$-mixing sources; this technical lemma, Lemma \ref{RPLessRE}, is the subject of Section~\ref{S8}. By use of Lemmas \ref{MainMainLemma} and \ref{RPLessRE}, the main theorem of this paper, Theorem \ref{MainTheorem}, the generalization of Theorem 1 of \cite{paperTogether} to $\psi$-mixing sources, can be stated and proved and this is done in Section~\ref{S9} . Application to this theorem to Markoff and order $m$ Markoff sources is stated and proved in Section~\ref{S10}. Some discussions are carried out in Section~\ref{Discuss} where in part, it is discussed, how to generalize Theorems 2, 3 and 4 of \cite{paperTogether} to $\psi$-mixing sources. Section~\ref{S12} discusses future research directions.

\section{Notation and definitions} \label{S3}

Let $X_1, X_2, \ldots, X_n, \ldots$, be a sequence of random variables defined on a probability space $(\Omega, \Sigma, P)$. The range of each $X_i$ is assumed to be a finite set $\mathbb X$. Denote this sequence of random variables by $X$. Such a sequence is called a source. Further discussion and assumption on the source will be carried out in Section 4.

%Let $\mathbb X$ be a finite set. Let time be discrete, $t = 1, 2, 3, \ldots$. Consider a stationary process $X = (X_t, t = 1, 2, 3, \ldots)$ where each $X_t$ is a random variable on $\mathbb X$. Note that we assume that the process is one-sided in time, that is runs from time $1$ to $\infty$, not $-\infty$ to $\infty$. Sigma field is the canonical sigma field on $\mathbb X^{\infty}$ corresponding to a process evolving in time on a finite set, see for example \cite{Shields}. Denote the probability distribution (which is stationary) on $\mathbb X^{\infty}$ corresponding to the process by $P$.

Sets will be denoted by latex mathbb notation, example, $\mathbb X, \mathbb Y$, and random variables by basic mathematical notation, for example $X, Y$. Sigma fields will be denoted by mathcal notation for example, $\mathcal S$.

The source space at each time, as stated before, is $\mathbb X$, and is assumed to be a finite set. The source reproduction space is denoted by $\mathbb Y$ which is assumed to be a finite set. Assume that $\mathbb X = \mathbb Y$.

$d: \mathbb X \times \mathbb Y \rightarrow [0, \infty )$ is the single-letter distortion measure. Assume that $d(x,x) = 0 \ \forall x \in \mathbb X$.

For $x^n \in \mathbb X^n, y^n \in \mathbb Y^n$, the $n$-letter rate-distortion measure is defined additively:
$$d^n(x^n, y^n) \triangleq \sum_{i=1}^n d(x^n(i), y^n(i))$$

where $x^n(i)$ denotes the $i^{th}$ component of $x^n$ and likewise for $y^n$.

$(X_1, X_2, \ldots, X_n)$ will be denoted by $X^n$.

A rate $R$ source-code with input space $\mathbb X$ and output space $\mathbb Y$ is a sequence $<e^n, f^n>_1^\infty$, where 
$$e^n: \mathbb X^n \rightarrow \{1, 2, \ldots, 2^{\lfloor nR  \rfloor} \}\quad \text{and}\quad f^n: \{1, 2, \ldots, 2^{\lfloor nR  \rfloor} \} \rightarrow \mathbb Y^n.$$

We say that rate $R$ is achievable for source-coding the source $X$  within distortion-level $D$ under the expected distortion criterion  if there exists a rate $R$ source code $<e^n, f^n>_1^\infty$ such that 
\begin{align} \label{EDDef}
\limsup_{n \to \infty} E \left [  \frac{1}{n} d^n(X^n, f^n(e^n(X^n))) \right ] \leq D
\end{align}

The infimum of all achievable rates under the expected distortion is an operational rate-distortion function, denoted by $R^E_X(D)$.

We say that rate $R$ is achievable for source-coding the source $X$  within distortion-level $D$ under the probability of excess distortion criterion  if there exists a rate $R$ source code $<e^n, f^n>_1^\infty$ such that 
\begin{align}\label{PDDef}
\lim_{n \to \infty} \Pr \left ( \frac{1}{n}  d^n(X^n, f^n(e^n(X^n))) > D \right ) = 0
\end{align}

The infimum of all achievable rates under the probability of excess distortion criterion  is an operational rate-distortion function, denoted by $R^P_X(D)$

We used $\limsup$ in (\ref{EDDef}) and $\lim$ in (\ref{PDDef}); in (\ref{PDDef}), we can equivalently use $\limsup$. This is because for a sequence of non-negative real numbers $a_n$, $\lim_{n \to \infty} a_n = 0$ is equivalent to $\limsup_{n \to \infty} a_n = 0$.

The block-independent approximation (henceforth shortened to BIA) $X^T$ source is a sequence of random vectors $(S_1, S_2, \ldots, S_n, \ldots)$, where $S_i$ are independent, and $\forall i$, $S_i \sim X^T$. To simplify notation, we will sometimes denote $(S_1, S_2, \ldots)$ by $S$. $S^n$ will denote $(S_1, S_2, \ldots, S_n)$. Note that BIA $X^T$ source is an i.i.d. vector source and will also be called the vector i.i.d. $X^T$ source.

The rate-distortion function for the vector i.i.d. $X^T$ source is defined in the same way as above; just that the source input space would be $\mathbb X^T$, the source output space will be $\mathbb Y^T$, the single letter distortion function would now be on $T$-length sequences and is defined additively, and when forming block-codes, we will be looking at blocks of $T$-length vectors. Details are as follows:

The source input space is $\mathbb X^T$. Denote it by $\mathbb S$. The source reproduction space is $\mathbb Y^T$. Denote it by $\mathbb T$. Denote a generic element of the source space by $s$ and that of the source reproduction space by $t$. Note that $s$ and $t$ are $T$-length sequences. Denote the $i^{th}$ component by $s(i)$ and $t(i)$ respectively.

The single letter distortion function, now, has inputs which are length $T$ vectors. It is denoted by $d_T$ and is defined additively using $d$ which has been defined before:

$d_T(s,t) \triangleq \sum_{i=1}^T d(s(i),t(i))$.

Note that $d_T$ is the same as $d^T$; just that we use superscript $T$ for $T$ length vectors, but now, we want to view a $T$-length vector as a scalar, and on this scalar, we denote the distortion measure by $d_T$.

$s^n$ will denote a block-length $n$ sequence of vectors of length $T$. Thus, $s^n(i)$, which denotes the $i^{th}$ component of $s^n$ is an element of $\mathbb K$. $s^n(i)(j)$ will denote the $j^{th}$ component of $s^n(i)$. 

The $n$-letter distortion function is defined additively using $d_T$:

For $s^n \in \mathbb S^n$, $t^n \in \mathbb T^n$,

$d_T^n(s^n,t^n) \triangleq \sum_{i=1}^nd_T(s^n(i), t^n(i))$.

When coding the vector i.i.d. $X^T$ source (for short, denoted by $S$), a rate $R$ source code is a sequence $<e^n, f^n>_1^\infty$, where $e^n: \mathbb S^n \rightarrow \{1, 2, \ldots, 2^{\lfloor nR  \rfloor} \}$ and $f^n: \{1, 2, \ldots, 2^{\lfloor nR  \rfloor} \} \rightarrow \mathbb T^n$.

We say that rate $R$ is achievable for source-coding the vector i.i.d. $X^T$ source within distortion-level $D$ under the expected distortion criterion  if there exists a rate $R$ source code $<e^n, f^n>_1^\infty$ such that 
\begin{align}
\lim_{n \to \infty} E \left [  \frac{1}{n} d_T^n(S^n, f^n(e^n(S^n))) \right ] \leq D
\end{align}

(Note that $S^n$ denotes $(S_1, S_2, \ldots, S_n)$).

The infimum of all achievable rates under the expected distortion criterion is the operational rate distortion function, denoted by $R^E_{X^T}(D)$. 

The information-theoretic rate-distortion function of the vector i.i.d. $X^T$ source is denoted and defined as
\begin{align}
R^I_{X^T}(D) \triangleq
\inf_{\mathbb T
 }
I(X^T; Y^T)
\end{align}
where  $\mathbb T$  is the set of $W: \mathbb S \rightarrow \mathbb P(\mathbb T)$ defined as
\begin{align}
\mathbb W \triangleq \left \{ W \ \left |\ \sum_{s \in \mathbb S, y \in \mathbb T}
p_{X^T}(s)W(t|s) d_T(s,t) \leq D  \right . \right \}
\end{align} 
where $p_{X^T}$ denotes the distribution corresponding to $X^T$.

Note that this is the usual definition of the information-theoretic rate-distortion function for an i.i.d. source; just that the source under consideration is vector i.i.d.

By the rate-distortion theorem, $R^E_{X^T}(D) = R^I_{X^T}(D)$.

Further, it is also known that 
\begin{align}
R^E_X(D) = \lim_{T \to \infty} \frac{1}{T} R^E_{X^T}(TD)
\end{align}

The channel is a sequence $c = <c^n>_1^\infty$ where
\begin{align}
c^n :& \mathbb X^n \rightarrow \mathbb P(\mathbb Y^n) \\
        & x^n \rightarrow c^n(\cdot|x^n)
\end{align}
When the block-length is $n$, the channel acts as $c^n(\cdot|\cdot)$; $c^n(y^n|x^n)$ is the probability that the channel output is $y^n$ given that the channel input is $x^n$. 

When the block-length is $n$, a rate $R$ deterministic channel encoder is a map $e_{ch}^{n}:\mathbb M_R^{n} \rightarrow \mathbb X^{n}$ and a rate $R$ deterministic channel decoder is a map $f_{ch}^{n}: \mathbb Y^{n} \rightarrow \hat {\mathbb M}_R^{n}$ where $\hat {\mathbb M}_R^{n} \triangleq \mathbb M_R^{n} \cup \{e\}$ is the message reproduction set where `e' denotes error. The encoder and decoder are allowed to be random in the sense that encoder-decoder is a joint probability distribution on the space of deterministic encoders and decoders. $<e_{ch}^{n}, f_{ch}^{n}>_1^\infty$ is the rate $R$ channel code.

Denote 
\begin{align}
g = <g^{n}>_1^\infty  \triangleq <e_{ch}^{n} \circ c^{n} \circ f_{ch}^{n}>_1^\infty
\end{align}
$g^n$ has input space $\mathbb M_R^{n}$ and output space $\hat {\mathbb M}_R^{n}$. Consider the set of channels
\begin{align}
\mathbb G_{\mathbb A} \triangleq \{ e \circ c \circ f \ | \ c \in \mathbb A \} 
\end{align}
$g \in \mathbb G_{\mathbb A}$ is a compound channel.  Rate $R$ is said to be reliably achievable over $g \in \mathbb G_{\mathbb A}$ if there exists a  rate $R$ channel code $<e_{ch}^{n}, f_{ch}^{n}>_1^\infty$  and a sequence $<\delta_n>_1^\infty$, $\delta_n \to 0$ as $n \to \infty$ such that 
\begin{align}
\sup_{m^{n} \in \mathbb M_R^{n}} g^{n}(\{ m^{n}\}^c|m^{n}) \leq \delta_n \ \forall c \in \mathbb A
\end{align}
Supremum of all achievable rates is the capacity of $c \in \mathbb A$. \emph{Note that this is the compound capacity, but will be referred to as just the capacity of $c \in \mathbb A$.}

The channel $c \in \mathbb A$ is said to communicate the source $X$ \emph{directly} within distortion $D$ if with input $X^{n}$ to $c^{n}$, the output is $Y^{n}$ (possibly depending on the particular $c \in \mathbb A$) such that
\begin{align} 
\Pr \left ( \frac{1}{n} d^{n}(X^{n}, Y^{n}) > D \right )  \leq \omega_n \forall c \in \mathbb A
\end{align}
for some $\omega_n \to 0$ as $n \to \infty$.

\section{Mixing condition used in this paper} \label{TheMix}

In this section, $\psi$-mixing processes are defined, properties of $\psi$-mixing processes are stated (and proved in the appendix) and intuition on $\psi$-mixing provided.

\subsection{Definition of $\psi$-mixing process}

Let $X_1, X_2, \ldots, X_n, \ldots$ be a sequence of random variables defined on a probability space $(\Omega, \Sigma, P)$. The random variables from $X_a$ to $X_b$ will be denoted by $X_a^b$, $1 \leq a \leq b \leq \infty$. The whole sequence $X_1^\infty$ will be denoted by $X^\infty$ or just by $X$. The range of each $X_i$ is assumed to be contained in a finite set $\mathbb X$. Note that time is assumed to be discrete. Note further, that it is assumed that the process is one-sided in time, that it runs from time $1$ to $\infty$, not $-\infty$ to $\infty$. The Borel sigma-field on $\mathbb X^\infty$ is defined in the standard way, and is denoted by $\mathcal F^\infty$; see Pages 1, 2 of \cite{Shields} for details.

$\mathbb X_a^b$ will denote the set corresponding to the $a^{th}$ to the $b^{th}$ coordinates of $\mathbb X^\infty$, $1 \leq a \leq b < \infty$. A sequence within these coordinates will be denoted by $x_a^b$, a random variable, by $X_a^b$. The Borel sigma-field on $\mathbb X_a^b$ is denoted by $\mathcal F_a^b$. Note that if $a$ and $b$ are finite, $\mathcal F_a^b = 2^{\mathbb X_a^b}$, the power set of $\mathbb X_a^b$.

For $\mathbb A \in \mathcal F_a^t$ and $\mathbb B \in \mathcal F_{t+\tau+1}^\infty$, we will have occasion to talk about the following probabilities:

\begin{align}
& \Pr( X_1^t \in \mathbb A)  \\
& \Pr(X_{t+\tau+1}^\infty \in \mathbb B) \nonumber \\
& \Pr(X_1^t \in \mathbb A, X_{t+\tau+1}^\infty \in \mathbb B) \nonumber
\end{align}

The intuitive meaning is clear: for example, $\Pr(X_1^t \in \mathbb A, X_{t+\tau+1}^\infty \in \mathbb B)$ refers to the probability that the random variable $X_1^t$ takes values in the set $\mathbb A$ and the random variables $X_{t+\tau+1}^\infty$ take values in the set $\mathbb B$. Mathematically, this is defined as follows. Define:
\begin{align}
& \mathbb A' = \{(a_1, a_2, \ldots, a_n, \ldots) | (a_1, a_2, \ldots, a_t) \in \mathbb A\}  \\
& \mathbb B' = \{(b_1, b_2, \ldots, b_n, \ldots) | (b_{t+\tau+1}, b_{t+\tau+2}, \ldots) \in \mathbb B\} \nonumber 
\end{align}

Then, 
\begin{align}
& \Pr(X_1^t \in \mathbb A) \triangleq P(X^\infty \in \mathbb A')  \\
& \Pr(X_{t+\tau+1}^\infty \in \mathbb B) \triangleq P(X_1^\infty \in \mathbb B') \nonumber \\
& \Pr(X_1^t \in \mathbb A, X_{t+\tau+1}^\infty \in \mathbb B) \triangleq P(X_1^\infty \in \mathbb A' \cap \mathbb B') \nonumber \\ 
\end{align}

Further, if $\Pr(X_1^t \in \mathbb A) > 0$, the following definition will be used:
\begin{align}\label{O17}
\Pr(X_{t+\tau+1}^{\infty}\in \mathbb{B}|X_1^t \in \mathbb{A}) \triangleq 
\frac{\Pr(X_1^t \in \mathbb A, X_{t+\tau+1}^\infty \in \mathbb B)}{\Pr(X_1^t \in \mathbb A)}
\end{align}

The one-sided version of $\psi$-mixing criterion of \cite{Bradley} will be used in this document, This is because the stochastic process under consideration in this document is one-sided in time, whereas the stochastic process under consideration in \cite{Bradley} is two-sided in time.

Define, for $\tau \in \mathbb W$, the set of whole numbers (non-negative integers),
\begin{multline}\label{O18}
	\psi(\tau)=\sup_{t \in \mathbb N}\sup_{\mathbb{A}\in \mathcal{F}_{1}^{t},\mathbb{B}\in \mathcal{F}_{t+\tau+1}^{\infty},\Pr(X_{1}^{t}\in \mathbb{A})>0,\Pr(X_{t+\tau+1}^{\infty}\in \mathbb{B})>0}\\
\left  | \frac{\Pr(X_1^t \in \mathbb A, X_{t+\tau+1}^{\infty}\in \mathbb{B})}{\Pr(X_1^t \in \mathbb{A})\Pr(X_{t+\tau+1}^{\infty}\in \mathbb{B})} -1 \right |
\end{multline}

The process $X$ is said to be $\psi$-mixing if $\psi(\tau) \to 0$ as $\tau \to \infty$.

The changes in (\ref{O18}) from \cite{Bradley} are:
\begin{itemize}
\item
 The first $\sup$ is taken over $t \in \mathbb Z$ in \cite{Bradley}, see Page 111 of \cite{Bradley}. Also, $t$ is denoted by $j$ in \cite{Bradley}. However, the sup in (\ref{O18}) is over $j \in \mathbb W$. This is because the process in \cite{Bradley} is two-sided in time, whereas we are considering a one-sided process. 
\item
A change of notation, where probabilities in (\ref{O18}) are written in terms of random-variables taking values in certain sets, whereas \cite{Bradley} considers the underlying probability space and writes probabilities of sets on that space, see Page 110, 111 of \cite{Bradley}. 
\item
The set $\mathbb A \in \mathcal F_1^t$ in (\ref{O18}), whereas if one used the deﬁnition in \cite{Bradley}, the set $\mathbb A$ would belong to $\mathcal F_{-\infty}^t$. This is, again, because the process in \cite{Bradley} is two-sided whereas the process in this paper is one-sided.
\end{itemize}

The reader is referred to \cite{Bradley} and 
\cite{Prohorov} 
for an overview of various kinds of mixing conditions. \cite{Bradley} gives a thorough overview of strong mixing conditions. \cite{Prohorov} mentions both weak mixing and strong mixing conditions though the coverage of strong mixing conditions is less thorough than in \cite{Bradley}.

Let $X$ be stationary. For $\mathbb B \subset \mathbb X^T$, denote the probability $P(X_{t+1}^{t+T} \in \mathbb B)$ (which is independent of $t$ since $X$ is stationary), by $P_T(\mathbb B)$. Note that $P_T$ is a probability distribution on $\mathbb X^T$ where the underlying sigma-field is the canonical sigma-field $2^{\mathbb X^T}$.

\subsection{Properties of $\psi$-mixing processes}

\begin{lemma} \label{L1}
Let $X$ be stationary, $\psi$-mixing. Then, $\forall t \in \mathbb N, \forall \tau \in \mathbb W, \forall T \in \mathbb W, \forall \mathbb A \subset \mathbb X^t, \forall \mathbb B \subset \mathbb X^T, P(X_1^t \in \mathbb A) > 0$,
\begin{align} \label{O19}
\Pr(X_{t+\tau+1}^{t+\tau+T}\in \mathbb{B}|X_1^t \in \mathbb{A}) = (1-\lambda_\tau)P_T(\mathbb B) + \lambda_\tau P'_{t,\tau,T,\mathbb A}(\mathbb B)
\end{align}
for some probability distribution $P'_{t, \tau, T, \mathbb A}$ on $\mathbb X^T$ (under the canonical sigma field on $\mathbb X^T$) which may depend on $t, \tau, T, \mathbb A$, and $\lambda_\tau \to 0$ as $\tau \to \infty$.
\end{lemma}
\begin{proof}
See Appendix \ref{Append}.
\end{proof}

\begin{lemma} \label{L2}
If $X$ is stationary, $\psi$-mixing, then $X$ is ergodic.
\end{lemma}
\begin{proof}
See Appendix \ref{Append}.
\end{proof}

\begin{lemma}\label{L3}
Let $X = (X_1, X_2, \ldots, X_n, \ldots)$ be a stationary, irreducible, aperiodic Markoff chain evolving on a finite set $\mathbb X$. Then, $X$ is $\psi$-mixing.
\end{lemma}
\begin{proof}
See Appendix \ref{Append}.
\end{proof}

Lemmas \ref{L2} and \ref{L3} have been proved in \cite{Bradley} for two-sided $\psi$-mixing processes. The proof of Lemma 3 uses the result from \cite{Bradley} on two-sided processes.

\begin{lemma} \label{L4}
Let $X = (X_1, X_2, \ldots)$ be a stationary, $\psi$-mixing process evolving on a set $\mathbb X$. For $L \in \mathbb N$, define $Z_t = X_{(t-1)L + 1}^{tL}$. Then, $Z=(Z_1, Z_2, \ldots)$ is a stationary, $\psi$-mixing process evolving on the set $\mathbb X^L$.
\end{lemma}
\begin{proof}
See Appendix \ref{Append}.
\end{proof}

\begin{lemma} \label{L5}
Let $X$ be a stationary, order $m$ Markoff chain evolving on a finite set $\mathbb X$. Define $Z_t = X_{(t-1)L+1}^{tL}$. Note that $Z=(Z_1, Z_2, \ldots)$ is a Markoff chain evolving on the set $\mathbb Z = \mathbb X^L$. Assume that $Z$ is irreducible, aperiodic. Then $X$ is $\psi$-mixing.
\end{lemma}
\begin{proof}
See Appendix \ref{Append}.
\end{proof}

It should be noted here, that a $\psi$-mixing process can have a rate of mixing as slow as is desired whereas a Markoff $\psi$-mixing chain implies exponential rate of convergence to the stationary distribution \cite{Bradley2}, \cite{Kesten}. Thus, the set of $\psi$-mixing processes is strictly larger than the set of Markoff or order $m$ Markoff chains. 

These lemmas are the same as the lemmas in \cite{Bradley} but for $1$-sided $\psi$-mixing processes, not $2$-sided $\psi$-mixing processes. Many of the proofs use the result from \cite{Bradley} for $2$-sided $\psi$-mixing processes and via a suitable construction, prove the same for $1$-sided $\psi$-mixing processes.

\subsection{Intuition on $\psi$-mixing}

Assume that $X$ is stationary. Note (\ref{O55}). $X_1^t$ and $X_{t+\tau+1}^\infty$ are independent if 
\begin{align}\label{O20}
P(X_{t+\tau+1}^\infty \in \mathbb B | X_1^t \in \mathbb A) = P(X_{t+\tau+1}^\infty \in \mathbb B) 
= P_T(\mathbb B)
\end{align}
Thus, (\ref{O55}) says that the process `becomes more and more independent' with time, further, this happens at a rate proportional to a factor $\lambda_\tau \to 0$ as $\tau \to \infty$ which is independent of the sets $\mathbb A$ and $\mathbb B$ in question, and also a multiplicative factor which depends on the probability of the set $\mathbb B$. This dependence on the probability of $\mathbb B$ is intuitively pleasing in the sense that, for example, if $P_T(\mathbb B) = 10^{-10}$ and $\lambda_{\tau} = 10^{-5}$, then without the multiplicative factor $P_T(\mathbb B)$, it says nothing meaningful; however, with the multiplicative factor $P_T(\mathbb B)$, it says something meaningful. A mixing condition can indeed be defined where $P_T(\mathbb B)$ does not exist on the right hand side in (\ref{O20}), this is the $\phi$-mixing criterion in \cite{Bradley}. An even weaker condition is the $\alpha$-mixing condition \cite{Bradley} where independence is measured in the sense of 
\begin{align}
P(\mathbb A \cap \mathbb B) = P(\mathbb A) P(\mathbb B)
\end{align}
instead of 
\begin{align}
P(\mathbb B|\mathbb A) = P(\mathbb B)
\end{align}

The $\phi$-mixing criterion has been used in the source coding literature, see for example \cite{7} and \cite{8}. In \cite{7}, it is proved that if a certain version of the goldwashing algorithm is applied to encode a stationary, $\phi$-mixing source, the expected distortion performance converges to the distortion-rate function of the source as the codebook length goes to $\infty$. In \cite{8}, it is proved that for sources which are $\phi$-mixing and have summable mixing coefficients, the redundancy of the fixed-database Lempel-Ziv algorithm with database size $n$ is lower bounded by a certain function of $n$ as described in \cite{8}.

\section{Idea of the proof} \label{S5}

Theorem 1 of \cite{paperTogether} will be generalized to $\psi$-mixing sources in this paper. This will be done by reducing the problem to the case when the source is i.i.d., and then, use Theorem 1 of \cite{paperTogether}.

The basic idea of the proof is the following: Choose $\tau, T$, where $\tau$ is ‘small’ compared to $T$. Denote $K_1 = X_1^T,  K_2 = X_{T+\tau+1}^{2T+\tau}, K_3 = X_{2T+2\tau+1}^{3T+2\tau}, \ldots$. Each $K_i$ has the same distribution; denote it by $K$. By Lemma 1, each $K_i$ has distribution close to $P_T$ in the sense of (\ref{O19}). Thus, $K_1, K_2, K_3, \ldots$,  is ‘close to’ an i.i.d. process. Theorem 1 from \cite{paperTogether} can be used and rates approximately
\begin{align}
\frac{T}{T+\tau} \frac{1}{T} R^E_K(TD)
\end{align}
are achievable for communication over a channel which is known to communicate the source $X$ to within a distortion $D$. Take $T \to \infty$ and it follows that rates $< R^E_X(D)$ are achievable, where $X$ is the $\psi$-mixing source. Finally, since the description of the channel is in terms of a probability of excess distortion criterion, we will prove that $R^P_X(D) \leq R^E_X(D)$ and this will prove that if a certain rate $R^E_X(\cdot)$ is achievable for the channel-coding problem, then so is the rate $R^P_X(D)$. 

A lot of technical steps are needed and this will be the material of the future sections. Note also, that there are various definitions of mixing in the literature which will make $K_1, K_2, \ldots$,  ‘almost’ independent, but the proof will not work for all these definitions. The definition of $\psi$-mixing is used primarily because (\ref{O19}) holds and this can be used to simulate the source $X$ in a way discussed in the next section, and this simulation procedure will be a crucial element of the proof.

\section{A simulation procedure for the stationary source $X$ which satisfies $\psi$-mixing} \label{S6}

By using Lemma 1, a procedure to simulate the source $X = (X_t, t = 1, 2, \ldots)$ will be described.

Fix $T$ and $\tau$, both strictly positive integers. Denote $n = (T+\tau)k$ for some strictly positive integer $k$.

We will generate a $(X'_1, X'_2, \ldots, X'_{(T+\tau)k})$, as described below.

First divide time into chunks of time $T$, $\tau$, $T$, $\tau$, $T$, $\tau$, and so on $\ldots$

Call these slots $A_1, B_1$, $A_2, B_2$, $\ldots$,  $A_i, B_i$, $\ldots$, $A_k, B_k$.

Thus,

$A_1$ contains ${X'}_1^T$.

$B_1$ contains ${X'}_{T+1}^{T+\tau}$.

$A_2$ contains ${X'}_{T+\tau+1}^{2T+\tau}$.

$B_2$ contains ${X'}_{2T+\tau+1}^{2T+2\tau}$.

$\vdots \ \ \ \vdots \ \ \ \vdots \ \ \ \vdots \ \ \ \vdots \ \ \ \vdots \ \ \ \vdots \ \ \ \vdots$

$A_i$ contains ${X'}_{(i-1)(T + \tau)+1}^{iT + (i-1)\tau}$.

$B_i$ contains ${X'}_{iT + (i-1)\tau+1}^{i(T+\tau)}$.

$\vdots \ \ \ \vdots \ \ \ \vdots \ \ \ \vdots \ \ \ \vdots \ \ \ \vdots \ \ \ \vdots \ \ \ \vdots$

$A_k$ contains ${X'}_{(k-1)(T + \tau)+1}^{kT + (k-1)\tau}$.

$B_k$ contains ${X'}_{kT + (k-1)\tau+1}^{k(T+\tau)}$.

Let $C_1 = 1$.

Generate $C_2, C_3, \ldots, C_k$ i.i.d., where $C_i$ is $1$ with probability $(1 - \lambda_{\tau})$ and $0$ with probability $\lambda_{\tau}$.

If $C_i = 1$, denote $A_i$ by $A_i^{(g)}$and if $C_i = 0$, denote $A_i$ by $A_i^{(b)}$. Think of superscript `g' as `good' and `b' as `bad'.

Generation of $(X'_1, X'_2, \ldots, X'_{(T+\tau)k})$ is carried out as follows: 

The order in which the $X'_i$s in the slots will be generated is the following: 

$A_1, A_2, B_1, A_3, B_2, \ldots, A_i, B_{i-1}, A_{i+1}, \ldots$.

Generate ${X'}_1^T$ (slot $A_1^{(g)}$) by the distribution $P_T$.

Assume that all $X_i$ have been generated until slot $A_{i-1}$, in other words, the generation in the following slots in the following order has happened: 

$A_1, A_2, B_1, A_3, B_2, \ldots, A_{i-1}, B_{i-2}$. 

The next two slots to be generated, as per the order stated above, is $A_i$ and then $B_{i-1}$. 

For slot $A_i$,

If it is a `g' slot, generate ${X'}_{(i-1)(T + \tau)+1}^{iT + (i-1)\tau}$ using $P_T$.

If it is a `b' slot, if $P(X_1^{(k-1)T+(k-2)\tau} = {x'}_1^{(k-1)T + (k-2)\tau}) > 0$, 
 generate ${X'}_{(i-1)(T + \tau)+1}^{iT + (i-1)\tau}$ using $P'_{t, \tau, T, \mathbb A}$ with $t = (k-1)T + (k-2)\tau$ and $\mathbb A =\linebreak \{{x'}_1^{(k-1)T + (k-2)\tau} \}$
where ${x'}_1^{(k-1)T + (k-2)\tau}$ is the simulated process realization so far. If $P(X_1^{(k-1)T+(k-2)\tau} = {x'}_1^{(k-1)T + (k-2)\tau}) > 0$, no process generation needs to be carried out anyway.

During the slot $B_{i-1}$, ${X'}_{(i-1)T + (i-2)\tau+1}^{(i-1)(T+\tau)}$ is generated using the probability measure $P$ of the stationary process given the values of the process already generated, that is, given ${x'}_{1}^{(k-1)T + (k-2)\tau}$ and ${x'}_{(i-1)(T + \tau)+1}^{iT + (i-1)\tau}$.

This finishes the description of the generation of the $(X'_1, X'_2, \ldots,\linebreak X'_{(T+\tau)k})$ sequence. 

Note that by Lemma 1 and the way the above simulation has been carried out,  $(X'_1, X'_2, \ldots, X'_{(T+\tau)k}) \sim (X_1, X_2, \ldots, X_{(T+\tau)k})$.

Note also, that during slots $A_i^{(g)}$, the source has distribution $X^T$ and is independent over these slots. This fact is of importance in the next section.

\section{The main lemma:  channel-coding theorem} \label{MainLemma}

\begin{lemma}\label{MainMainLemma}
%\emph{Lemma 1:}

Let $c = <c^n>_1^\infty$ directly communicate the source $X$, assumed to be $\psi$-mixing, within distortion $D$.

Let $\lambda > 0$ (think of $\lambda$ small; $\lambda << 1$). Choose $\beta > 0$ (think of $\beta$ small; $\beta << 1 - \lambda$). Choose $\tau$ large enough so that $\lambda_{\tau} \leq \lambda$. Then, rates
\begin{align}
R < \frac{1-\lambda_{\tau}-\beta}{T+\tau}R^E_{X^T} \left ( \frac{(T + \tau)D}{1 - \lambda_{\tau} - \beta} \right )
\end{align}
are reliably achievable over $c$ $\forall T \geq 1$ (think of $T$ large).
\end{lemma}

%\emph{Proof:}
\begin{proof}

Choose $T \geq1$.

Let $n = (T+\tau)k$ for some large $k$. $n$ is the block-length.

Generate $C_1, C_2, \ldots$ as described previously.

Generate $2^{\lfloor nR \rfloor}$ codewords of block-length $(T+\tau)k = n$ by use of the simulation procedure described previously. Note that $C_1, C_2, \ldots$ is the same for generating all the $2^{\lfloor nR \rfloor}$ codewords.

Note that over $A_i^{(g)}$ time slots, the codewords are generated i.i.d., as in Shannon's random-coding argument; this generation during $A_i^{(g)}$ is done i.i.d. $X^T$.

Recall the behavior of the channel which directly communicates the source $X$  within distortion $D$. End-to-end,
\begin{align}
\lim_{n \to \infty} \Pr \left ( \frac{1}{n} d^n(X^n, Y^n) > D \right ) = 0
\end{align}
Let us look at the behavior of the channel restricted to time slots $A_i^{(g)}$.

Assume that the fraction of `g' slots among the $k$ $A_i$ slots is $\geq 1 - \lambda_{\tau} - \beta$. That is, number of $A_i^{(g)}$ slots is larger than or equal to  $\lfloor (1 - \lambda_{\tau} - \beta)k \rfloor + 1$. Denote $N = \lfloor (1 - \lambda_{\tau} - \beta)k \rfloor + 1$. This is a high probability event and the probability $\to 1$ as $k \to \infty$ for any $\beta$. If this even does not happen, we will declare decoding error; hence, in what follows, assume that this is the case.

Restrict attention to the first $N$  $A_i^{(g)}$ slots. Rename these slots $G_1$, $G_2$, $\ldots$, $G_N$.

Denote the part of the source during slot $G_i$ by $S_i$. Note that $S_i$ is a $T$-length vector.

Denote $S = (S_1, S_2, \ldots, S_N$).

Denote the channel output during slot $G_i$ by $T_i$. Note that $T_i$ is a $T$-length vector. Denote $T = (T_1, T_2, \ldots, T_N)$.

Recall the definition of the distortion function $d_T$ for $T$-length vectors, and its $n$-block additive extension.

Over $G_i$ slots, then, 
\begin{align}
\lim_{N \to \infty}
\Pr \left ( \frac{1}{N} \sum_{i=1}^N
d_T(S_i, T_i) 
> \frac{(T+\tau)kD}{N}
\right ) 
= 0
\end{align}

By substituting $N = \lfloor (1 - \lambda_{\tau} - \beta)k \rfloor + 1$, it follows, after noting that
\begin{align}
\frac{k}{\lfloor (1 - \lambda_{\tau} - \beta)k \rfloor + 1} \leq \frac{1}{1 - \lambda_{\tau} - \beta}
\end{align} 
that
\begin{align}
\lim_{k \to \infty}
\Pr \left ( 
\frac{1}{\lfloor (1 - \lambda_{\tau} - \beta)k \rfloor + 1}
\sum_{i=1}^{\lfloor (1 - \lambda_{\tau} - \beta)k \rfloor + 1}
d_T(S_i, T_i) 
> 
\frac{(T+\tau)D}{1-\lambda_{\tau}-\beta}
\right )
= 0
\end{align}

Recall again that $S_i$ are i.i.d. $X^T$ and that, codeword generation over $G_i$ slots is i.i.d. 

We have reduced, then, the problem to that where it is known that an i.i.d. source is directly communicated over a channel within a certain probability of excess distortion and we want to calculate a lower bound on the capacity of the channel -- this is Theorem 1 of \cite{paperTogether}. 

If each $G_i$ is considered to be a single unit of time, or in other words, over $G_i$, the uses of the channel is considered as a single channel use, we are thus able, by use of Theorem 1 of \cite{paperTogether} to communicate at rates
\begin{align}
R < R^E_{X^T}\left ( \frac{(T+\tau)D}{1-\lambda_{\tau}-\beta }\right ) \ \mbox{(per channel use)}
\end{align}
Total time of communication, though, has been $(T+\tau)k$ and there are $\lfloor (1 - \lambda_{\tau} - \beta)k \rfloor + 1$ $G_i$ slots over which the communication takes place. Noting that 
\begin{align}
\frac{\lfloor (1 - \lambda_{\tau} - \beta)k \rfloor + 1}{(T+\tau)k} \geq 
                             \frac{1 - \lambda_{\tau} - \beta}{(T+\tau)}
\end{align}
it follows that rates
\begin{align}
R < \frac{1 - \lambda_{\tau} - \beta}{(T+\tau)} R^E_{X^T}\left ( \frac{(T+\tau)D}{1-\lambda_{\tau}-\beta }\right )
\end{align}
are achievable for reliable communication over  the original channel $c$ per channel use of $c$. 
\end{proof}

Roughly, the details of codebook generation and decoding are as follows:

Let reliable communication be desired at a rate $R$ which is such that there exist $\tau, \beta, T$ such that  
\begin{align}
R < \frac{1-\lambda_{\tau}-\beta}{T+\tau}R^E_{X^T} \left ( \frac{(T + \tau)D}{1 - \lambda_{\tau} - \beta} \right )
\end{align}

Generate $C_1, C_2, \ldots$. Assume that this knowledge is available at both encoder and decoder

Generate $2^{\lfloor k(T + \tau) R \rfloor}$ codewords using the simulation procedure.

If the number of `g' slots is less than $\lfloor (1-\lambda_{\tau}-\beta)k \rfloor $, declare error.

Else, restrict attention only the first $\lfloor (1-\lambda_{\tau}-\beta)k \rfloor $ $A_i^{(g)}$ slots which have been renamed $G_1, G_2, \ldots$.

Over these slots, the codebook generation is i.i.d., and then, use the procedure from Theorem 1 of \cite{paperTogether}.

\section{$R^P_X(D) \leq R^E_X(D)$ if $X$ is stationary and satisfies $\psi$-mixing} \label{S8}

\begin{lemma}\label{RPLessRE}
%\emph{Lemma 2}
Let $X = (X_t, t = 1, 2, 3, \ldots)$ be stationary process which satisfies $\psi$-mixing. Then, $R^P_X(D) \leq R^E_X(D)$.
\end{lemma}

%\emph{Proof:}
\begin{proof}
By Lemma \ref{L2}, $X$ is ergodic. Thus, $X$ is stationary, ergodic.

The proof now, relies on \cite{GallagerInformationTheory}, Pages 490-499, where the rate-distortion theorem is proved for stationary, ergodic sources.

First, note the notation in \cite{GallagerInformationTheory}. \cite{GallagerInformationTheory} defines $R_L(D)$ and $R(D)$, both on Page 491. Note that by the rate-distortion theorem for an i.i.d. source, it follows that
\begin{align}
R_L(D) \ \mbox{(notation in \cite{GallagerInformationTheory})}  = \frac{1}{T} R^E_{X^T}(TD) \ \mbox{(our notation)}
\end{align}
Thus, 
\begin{align} \label{H}
	R(D) \ \mbox{(notation in \cite{GallagerInformationTheory})} &= \lim_{T \to \infty}  R^E_{X^T}(TD) \ \mbox{(our notation)}\\
	\notag &= R^E_X(D) \ \mbox{(our notation)}
\end{align}

Look at Theorem 9.8.2 of \cite{GallagerInformationTheory}. This theorem holds if probability of excess distortion criterion is used instead of the expected distortion criterion: see (9.8.10) of \cite{GallagerInformationTheory}. By mapping the steps carefully, it follows that rate $R_1(D - \epsilon)$ (notation in \cite{GallagerInformationTheory}) is achievable for source-coding the source $X$ under a probability of excess distortion $D$ for all $\epsilon > 0$. Note that it follows that rates $R_1(D - \epsilon)$ are achievable, not necessarily rates $R_1(D)$. This is because in (9.8.10), when making further arguments, $\hat{d}$ is made $D + \frac{\delta}{2}$ and not $D$. Hence, we need to keep a distortion level smaller than $D$ in $R_1(\cdot)$ to make this rate achievable for the probability of excess distortion criterion. Next, we construct the $L^{th}$ order super source as described on Page 495 of \cite{GallagerInformationTheory}: Define ${X'}^{t} = X_{(t-1)L+1}^{tL}$. Then, $X' =  ({X'}^{t}, t = 1, 2, 3, \ldots)$ is the $n^{th}$ order super-source. $X'$ is stationary, $\psi$-mixing because $X$ is (Lemma \ref{L4}), and thus, stationary, ergodic, by Lemma~\ref{L2}. One can thus use Theorem 9.8.2 of \cite{GallagerInformationTheory} again to argue that rate $R_L(D-\epsilon)$ (notation of \cite{GallagerInformationTheory}) is achievable for source-coding the source $X$ under a probability of excess distortion $D$ for all $\epsilon > 0$. By taking a limit as $L \to \infty$ (the limit exists by Theorem 9.8.1 in \cite{GallagerInformationTheory}), it follows that rate $R(D-\epsilon)$ (notation in \cite{GallagerInformationTheory}) is achievable for source-coding the source $X$ under a probability of excess distortion $D$ for all $\epsilon > 0$. As stated at the end of the proof of Theorem 9.8.1 in \cite{GallagerInformationTheory}, $R(D)$ is a continuous function of $D$. Thus, it follows that rates $<R(D)$ are achievable for source-coding the source $X$ under a probability of excess distortion $D$. At this point, the lemma follows from (\ref{H}).
\end{proof}

\section{Generalization of Theorem 1 in Part I to stationary sources satisfying $\psi$-mixing} \label{S9}

Before we prove the theorem, note the following: Let $f: [0, \infty) \rightarrow [0, \infty )$ be a convex $\cup$ non-increasing function. Let $f(0) = K$. Let $0 < a < a'$. Then,
\begin{align}
|f(a) - f(a')| \leq \frac{K}{a}(a'-a)
\end{align}

\begin{thm}\label{MainTheorem}
Let $c$ be a channel over which the source $X$, assumed to be stationary, $\psi$-mixing, is directly communicated within probability of excess distortion $D$, $D > 0$. Then, rates $<R^P_{X}(D)$ are reliably achievable over $c$.

\end{thm}
\begin{proof}
Since $R^P_X(D) \leq R^E_X(D)$ by Lemma \ref{RPLessRE} and since it is known that 
\begin{align}
R^E_X(D) = \lim_{T \to \infty} \frac{1}{T} R^E_{X^T}(TD) 
\end{align}
it is sufficient to prove that rates less than
\begin{align}
\lim_{T \to \infty} \frac{1}{T} R^E_{X^T}(TD)
\end{align}
are reliably achievable over $c$.

To this end, denote
\begin{align}
D' \triangleq \frac{D}{1-\lambda_{\tau}-\beta}
\end{align}

Then, 
\begin{align}
       &  \frac{1 - \lambda_{\tau} - \beta}{T+\tau}  R^E_{X^T}  ((T+\tau)D') - 
      \lim_{T \to \infty} \frac{1}{T} R^E_{X^T}(TD') \\
=\,    & \frac{1 - \lambda_{\tau} - \beta}{T+\tau}  R^E_{X^T}  ((T+\tau)D')  - 
                                                                    \frac{1}{T+\tau}  R^E_{X^T} ((T+\tau)D')  \label{T1} \\
      & + \frac{1}{T+\tau}  R^E_{X^T} ((T+\tau)D') - \frac{1}{T}  R^E_{X^T} ((T+\tau)D') \label{T2} \\
      & + \frac{1}{T}  R^E_{X^T} ((T+\tau)D') - \frac{1}{T}  R^E_{X^T} (TD')  \label{T3} \\
          & + \frac{1}{T}  R^E_{X^T} (TD') -  \lim_{T \to \infty} \frac{1}{T} R^E_{X^T}(TD') \label{T4}
\end{align} 

Expression in (\ref{T1}) is 
\begin{align}
 \frac{-\lambda_{\tau} - \beta}{T+\tau}  R^E_{X^T}  ((T+\tau)D')
\end{align}
Note that 
\begin{align}
R^E_{X^T}  ((T+\tau)D') \leq T \log | \mathbb X|
\end{align}
Thus, the absolute value of the expression in (\ref{T1}) is upper bounded by $(\lambda_{\tau} + \beta) \log|\mathbb X|$.

Expression in (\ref{T2}) is
\begin{align}
\frac{-\tau}{T} \left ( \frac{1}{T+\tau} R^E_{X^T}  ((T+\tau)D') \right )
\end{align}
Note that 
\begin{align}
R^E_{X^T}  ((T+\tau)D') \leq T\log |\mathbb X|
\end{align} 
It then follows that expression in (\ref{T2}) $\to 0$ as $T \to \infty$.

 Expression in (\ref{T3}) is 
 \begin{align}
 \frac{1}{T}  R^E_{X^T} \left ( T(D' + \frac{\tau}{T}D') \right ) - \frac{1}{T}  R^E_{X^T} (TD')
 \end{align}
 $\frac{1}{T}R^E_{X^T} (TD)$ is a convex $\cup$ non-negative function of $D$, upper bounded by $\log |\mathbb X|$. It follows that 
  \begin{align}
&\quad\  \frac{1}{T}  R^E_{X^T} \left ( T(D' + \frac{\tau}{T}D') \right ) - \frac{1}{T}  R^E_{X^T} (TD') \nonumber \\
 &\leq \frac{\log|\mathbb X|}{D'} \left ( (D' + \frac{\tau}{T}D') - D'\right )  \nonumber \\
 & \to 0 \ \mbox{as} \ T \to \infty
 \end{align}

 Expression in (\ref{T4}) $\to 0$ as $T \to \infty$.
 
 By noting the bound on the absolute value of expression  (\ref{T1}) proved above and by noting, as proved above, that expressions in (\ref{T2}), (\ref{T3}), and (\ref{T4}) $\to 0$ as $T \to \infty$, it follows that $\exists \ \epsilon_T \to 0$ as $T \to \infty$, possibly depending on $\lambda_{\tau}$ and $\beta$ such that  
 \begin{align}
	 \left |  \frac{1 - \lambda_{\tau} - \beta}{T+\tau}  R^E_{X^T}  ((T+\tau)D') -  \lim_{T \to \infty} \frac{1}{T} R^E_{X^T}(TD') \right | \leq (\lambda_{\tau} + \beta)|\mathbb X| + \epsilon_T\hspace{-1ex}
 \end{align}

 By Lemma \ref{MainMainLemma}, and by recalling that $D' = \frac{D}{1 - \lambda_{\tau} - \beta}$ it follows that rates less than
 \begin{align}
 \lim_{T \to \infty} \frac{1}{T} R^E_{X^T}\left ( T\frac{D}{1 - \lambda_{\tau} - \beta}\right )
 - (\lambda_{\tau} + \beta)|\mathbb X| - \epsilon_T
 \end{align}
 are achievable reliably over $c$.

By using the fact that $\lambda_{\tau}$ and $\beta$ can be made arbitrarily small and $\epsilon_T \to 0$ as $T \to \infty$, and that, the function  \begin{align}
 \lim_{T \to \infty} \frac{1}{T} R^E_{X^T}(TD) 
 \end{align}
 is continuous  in $D$, it follows that rates less than 
 \begin{align}
 \lim_{T \to \infty} R^E_{X^T}(TD)
 \end{align}
 are reliably achievable over $c$ from which, as stated at the beginning of the proof of this theorem, it follows that rates less than $R^P_{X}(D)$ are reliably achievable over $c$. 
 \end{proof}
 
Note that statements concerning resource consumption have not been made either in Theorem \ref{MainTheorem} or Lemma \ref{MainMainLemma} in this paper whereas they are part of Theorem 1 in \cite{paperTogether}. For the corresponding statements concerning resource consumption, see Section \ref{Discuss}. Further, the way Theorem \ref{MainTheorem} or Lemma \ref{MainMainLemma} are stated in this paper, the channel does not belong to a set whereas in Theorem 1 in \cite{paperTogether}, the channel may belong to a set. For the corresponding statement where the channel may belong to a set, see Section \ref{Discuss}.

 \section{Application to Markoff chains and order $m$\\ Markoff chains} \label{S10}
 
 Let $X = (X_t, t = 1, 2, \ldots)$ be a stationary, irreducible, aperiodic Markoff chain evolving on a 
 finite set $\mathbb X$.  By Lemma \ref{L3}, $X$ is $\psi$-mixing. $X$ is thus, stationary, $\psi$-mixing and thus, Theorem \ref{MainTheorem} holds for stationary, irreducible Markoff chains evolving on a finite set.

Let $X = (X_i, i \in \mathbb N)$ be an order $m$ stationary Markoff chain. Define $Z_i = X_{(i-1)m+1}^{im}$. Then, $Z = (Z_i, i \in \mathbb N)$ is a Markoff chain. By Lemma \ref{L4}, $Z$ is stationary. Assume that this $Z$ is irreducible, aperiodic. By Lemma \ref{L5}, $X$ is $\psi$-mixing, and thus, Theorem \ref{MainTheorem} holds.

\section{Discussion} \label{Discuss}

It is really (\ref{O19}) that is crucial to the proof, not that $\psi$-mixing criterion; this is because it is (\ref{O19}) which is needed for carrying out the simulation procedure described in Section \ref{S6}. Other places where $\psi$-mixing criterion is used in minor ways is  to prove ergodicity and some other properties needed to finish parts of the proof but it is possible that they can be proved by use of (\ref{O19}) too (or can just be taken as assumptions). However, the assumption of $\psi$-mixing suffices, and since this condition holds for Markoff and order $m$ Markoff sources (under stationarity, irreducibility, aperiodicity assumptions as stated above), the theorem has been proved for quite a large class of sources.

In Theorem 1 of \cite{paperTogether}, the channel may belong to a set whereas the way Lemma \ref{MainMainLemma} and Theorem \ref{MainTheorem} are stated in this paper, the channel does not belong to a set. However, it is easy to see that the proof of Lemma \ref{MainLemma} does not require knowledge of the channel transition probability; only the end-to-end description that the channel communicates the source to within the distortion level is needed; for this reason, Theorem \ref{MainTheorem} in this paper generalizes to the case when the channel belongs to a set for the same reason as \cite{paperTogether}. A source-channel separation theorem has also been stated and proved in Theorem 2 in \cite{paperTogether}; this can be done in this paper too. Statements concerning resource consumption have not been made in this paper in Lemma \ref{MainMainLemma} or Theorem \ref{MainTheorem}. They follow for the same reason as in \cite{paperTogether}: in this context, note that the codebook in the proof of Lemma \ref{MainLemma} consists of codewords which are independent of each other and further,  each codeword has the distribution as the process $X$; this point is the only observation needed to prove the statements concerning resource consumption. Finally, generalization to the unicast, multi-user setting, namely Theorem 3 and 4 of \cite{paperTogether} follow for the same reason as in \cite{paperTogether}. In this context, the only observation that needs to be made is the same as above that the codewords in the proof of Lemma \ref{MainLemma} follow the distribution of the process $X$.

\section{Future research directions} \label{S12}

\begin{itemize}
\item
Generalize Theorem \ref{MainTheorem} to arbitrary stationary, ergodic processes, not just those which satisfy $\psi$-mixing, to the extent possible.
\item
In particular, explore a generalization to $B$-processes \cite{Gray}, the closure of the set of Markoff chains of finite order.
\item
Consider an alternate proof strategy for proving Theorem \ref{MainTheorem} which uses methods from classical ergodic and rate-distortion theory, that is, methods similar to, for example, \cite{GallagerInformationTheory} and \cite{Gray}, and thus, does not rely on the decomposition (\ref{O19}). This might help prove Theorem \ref{MainTheorem} for general stationary, ergodic sources, not just those which satisfy $\psi$-mixing.
\item
Further, consider a strategy based on the theory of large deviations, in the first instance, for irreducible, aperiodic Markoff chain source. For i.i.d. sources, a large deviations based method was indeed used in Part 1 \cite{paperTogether}.
\item
Generalize Theorem \ref{MainTheorem} to stationary, ergodic sources which evolve continuously in space and time (some assumptions might be needed on the source). Since only the end-to-end description of the channel as communicating the source $X$ within distortion level $D$ is used and not the exact dynamics of the channel, the proof given in Part 1 for Theorems 2 and 4, and for similar theorems in this paper, directly holds for channels which evolve continuously in space and time. The channel $k = <k^n>_1^\infty$ would however need to be rigorously defined for continuous time evolution. Further, the encoder-decoder $<e^n, f^n>_1^\infty$ would need to be defined on appropriate spaces so that the interconnection $<e^n \circ k^n \circ f^n>_1^\infty$ makes sense.
\item
Research the possibility of an operational rate-distortioon theory for stationary, ergodic sources (satisfying other conditions). An operational theory for i.i.d. sources has been presented in \cite{MukulPhDThesis}.
\item
The channel has been assumed to belong to a set in Part I \cite{paperTogether} and the same is the case in this paper.  However, the source is assumed to be known. Research the generalization of results in this paper to compound sources.
\end{itemize}

\section{Acknowledgements}
The authors are extremely grateful to Prof. Robert Gray for his time and many insightful discussions. The authors also thank Prof. Richard Bradley for many important e-mail conversations which helped shed more light on the $\psi$-mixing criterion.

\appendix

%24

\section{ Proofs of properties of $\psi$-mixing sequences} \label{Append}
Proof of Lemma \ref{L1} :
\begin{proof}
From (\ref{O18}) and (\ref{O17}), it follows that $\psi(\tau)$ can be alternatively be written as
\begin{multline} \label{O54}
\psi(\tau)=\sup_{t \in \mathbb N}\sup_{\mathbb{A}\in \mathcal{F}_{1}^{t},\mathbb{B}\in \mathcal{F}_{t+\tau+1}^{\infty},\Pr(X_{1}^{t}\in \mathbb{A})>0,\Pr(X_{t+\tau+1}^{\infty}\in \mathbb{B})>0}\\ 
\left  | \frac{\Pr(X_{t+\tau+1}^{\infty}\in \mathbb{B}|X_1^t \in \mathbb{A})}{\Pr(X_{t+\tau+1}^{\infty}\in \mathbb{B})} -1 \right |
\end{multline}
From (\ref{O54}), it follows that $\exists\lambda_{\tau}\rightarrow 0$ as $\tau\rightarrow\infty$ such that $\forall t\in \mathbb{N}$, $\forall\tau\in \mathbb{W}$, $\forall \mathbb{A}\in \mathcal{F}_{1}^{t}, \forall \mathbb{B}\in \mathcal{F}_{t+\tau+1}^{\infty}, \Pr(X_{1}^{t}\in \mathbb{A})>0, \Pr(X_{t+\tau+1}^{\infty}\in \mathbb{B})>0,$
\begin{align}\label{O55}
|\Pr(X_{t+\tau+1}^{\infty}\in \mathbb{B}|X_{1}^{t}\in \mathbb{A})-\Pr(X_{t+\tau+1}^{\infty}\in \mathbb{B})|\leq\lambda_{\tau}\Pr(X_{t+\tau+1}^{\infty}\in \mathbb{B})
\end{align}
From (\ref{O55}), it follows tha $\exists\lambda_{\tau}\rightarrow 0$ as $\tau\rightarrow\infty$ such that $\forall t\in \mathbb{N}$, $\forall\tau\in \mathbb{W}$, $\forall \mathbb{A}\in \mathcal{F}_{1}^{t}, \forall \mathbb{B}\in \mathcal{F}_{t+\tau+1}^{\infty}, \Pr(X_{1}^{t}\in \mathbb{A})>0, \Pr(X_{t+\tau+1}^{\infty}\in \mathbb{B})>0,$
\begin{align}\label{O56}
(1-\lambda_{\tau})\Pr(X_{t+\tau+1}^{\infty}\in \mathbb{B})\leq \Pr(X_{t+\tau+1}^{\infty}\in \mathbb{B}|X_{1}^{t}\in \mathbb{A})
\end{align}
Specializing (\ref{O56}), it follows that,
\begin{align}\label{O57}
(1-\lambda_{\tau})\Pr(X_{t+\tau+1}^{t+\tau+T}\in \mathbb{B})\leq \Pr(X_{t+\tau+1}^{t+\tau+T}\in \mathbb{B}|X_{1}^{t}\in \mathbb{A})
\end{align}
$\forall t\in \mathbb{N}$, $\forall\tau\in \mathbb{W}$, $\forall T\in \mathbb{W}$, $\forall \mathbb{A}\subset \mathbb{X}^{t}$, $\forall \mathbb{B}\subset \mathbb{X}^{T}$, $\Pr(X_{1}^{t}\in \mathbb{A})>0$, $\Pr(X_{t+\tau+1}^{t+\tau+T}\in \mathbb B) >0.$

Note that $\Pr(X_{t+\tau+1}^{t+\tau+T})=P_{T}(\mathbb B)$ . Substituting this into (\ref{O57}), it follows that $\forall t\in \mathbb{N}$, $\forall\tau\in \mathbb{W}$, $\forall T\in \mathbb{W}$, $\forall \mathbb{A}\subset \mathbb{X}^{t}$, $\forall \mathbb{B}\subset \mathbb{X}^{T}$, $\Pr(X_{1}^{t}\in \mathbb{A})>0$, $\Pr(X_{t+\tau+1}^{t+\tau+T}\in \mathbb B) >0,$
\begin{align} \label{O58}
(1-\lambda_{\tau})P_{T}(\mathbb B)\ \leq \Pr(X_{t+\tau+1}^{t+\tau+T}\in \mathbb{B}|X_{1}^{t}\in \mathbb{A})
\end{align}
If $\lambda_{\tau}=0$, it follows from (\ref{O55}), that for $\Pr(X_{t+\tau+1}^{t+\tau+T}\in \mathbb{B})>0, P(X_{1}^{t}\in \mathbb{A})>0,$
\begin{align} \label{O59}
\Pr(X_{t+\tau+1}^{t+\tau+T}\in \mathbb{B}|X_{1}^{t}\in \mathbb{A})=(1-\lambda_{\tau})P_{T}(\mathbb B)
\end{align}
and the above equation also holds if $P_{T}(\mathbb{B})=0$ but $P(X_{1}^{t}\in \mathbb{A})>0$; thus, (\ref{O19}) holds with any probability distribution $P_{t,\tau,T,\mathbb{A}}'$ on $\mathbb{X}^{T}.$

If $\lambda_{\tau}>0$, define
\begin{align}\label{O60}
 P_{t,\tau,T,\mathbb{A}}'(\displaystyle \mathbb{B})=\frac{P(X_{t+\tau+1}^{t+\tau+T}\in \mathbb{B}|X_{1}^{t}\in \mathbb{A})-(1-\lambda_{\tau})P_{T}(\mathbb{B})}{\lambda_{\tau}}
\end{align}
From (\ref{O58}) , it follows that $\forall t\in \mathbb{N}$, $\forall\tau\in \mathbb{W}$, $\forall T\in \mathbb{W}$, $\forall \mathbb{A}\in \mathbb{X}^{t}$, $\forall \mathbb{B}\in \mathbb{X}^{T}$, $P(X_{1}^{t}\in \mathbb{A})>0$,  $P(X_{t+\tau+1}^{t+\tau+T}\in \mathbb{B})>0$, $\exists \lambda_\tau \to 0$ as $\tau \to \infty$ such that 
\begin{align}\label{O61}
P(X_{t+\tau+1}^{t+\tau+T}\in \mathbb{B}|X_{1}^{t}\in \mathbb{A})=(1-\lambda_{\tau})P_{T}(\mathbb{B})+\lambda_{\tau}P_{t,\tau,T,\mathbb{A}}'(\mathbb{B})
\end{align}
for some probability distribution $P_{t,\tau,T,\mathbb{A}}'$ on $\mathbb{X}^{T}$ which may depend on $t, \tau,\allowbreak T, \mathbb A.$

Finally, note that if $P_{T}(\mathbb B) =0$, (\ref{O19}) still holds with definition (\ref{O60}) for $P_{t,\tau,T,\mathbb{A}}'$ since all the three probabilities in question are individually zero.

This finishes the proof of the lemma.
\end{proof}

Proof of Lemma \ref{L2}:
\begin{proof}
In order to prove this lemma, it is sufficient to prove the condition on Page 19 in \cite{Shields} (which implies ergodicity as is proved on the same page of \cite{Shields}), and which can be re-stated as
\begin{align} \label{O62}
\lim_{N\rightarrow\infty}\frac{1}{N}\sum_{\tau=0}^{N-1}P(X_{1}^{t}=a_{1}^{t},\ X_{\tau+1}^{\tau+T}=b_{1}^{T})=P(X_{1}^{t}=a_{1}^{t})P(X_{1}^{T}=b_{1}^{T})
\end{align}
$\forall t\in \mathbb{N}, \forall T\in \mathbb{N}, \forall a_{1}^{t}\in \mathbb{X}^{t}, \forall b_{1}^{T}\in \mathbb{X}^{T}.$

To this end, note, first, that from (\ref{O55}), it follows that $\exists\lambda_{\tau}\rightarrow 0$ as\linebreak $\tau\rightarrow\infty$ such that $\forall t\in \mathbb{N}, \forall \mathbb{A}\in \mathcal{F}_{1}^{t}, \forall \mathbb{B}\in \mathcal{F}_{t+\tau+1}^{\infty}, P(X_{1}^{t}\in \mathbb{A})>0, P(X_{t+\tau+1}^{\infty}\in \mathbb{B})>0,$
\begin{align} \label{O63}
	(1-\lambda_{\tau})P(X_{t+\tau+1}^{\infty}\in \mathbb{B})&\leq P(X_{t+\tau+1}^{\infty}\in \mathbb{B}|X_{1}^{t}\in \mathbb{A})\\
	\notag &\leq(1+\lambda_{\tau})P(X_{t+\tau+1}^{\infty}\in \mathbb{B})
\end{align}
Thus, $\exists\lambda_{\tau}\rightarrow 0$ as $\tau\rightarrow\infty$ such that $\forall t\in \mathbb{N}, \forall \mathbb{A}\in \mathcal{F}_{1}^{t}, \forall \mathbb{B}\in \mathcal{F}_{t+\tau+1}^{\infty},  P(X_{1}^{t}\in$ A) $>0, P(X_{t+\tau+1}^{\infty}\in \mathbb{B})>0,$
\begin{align}\label{O64}
	&\quad\ (1-\lambda_{\tau})P(X_{1}^{t}\in \mathbb{A})P(X_{t+\tau+1}^{\infty}\in \mathbb{B}) \nonumber \\
	&\leq P(X_{1}^{t}\in \mathbb{A},\ X_{t+\tau+1}^{\infty}\in \mathbb{B}) \nonumber \\
	&\leq(1+\lambda_{\tau})P(X_{1}^{t}\in \mathbb{A})P(X_{t+\tau+1}^{\infty}\in \mathbb{B})
\end{align}
%26
If $P(X_{1}^{t}=a_{1}^{t})=0$, then both the left hand side and the right hand side in (\ref{O62}) are zero. If $P(X_{1}^{T}=b_{1}^{T})=0$, by use of the assumption that $X$ is stationary and thus noting that $P(X_{\tau+1}^{\tau+T})=P(X_{1}^{T})$ , it follows that both the left hand side and the right hand side in (\ref{O62}) are zero. If neither $P(X_{1}^{t}= a_{1}^{t})=0$ nor $P(X_{1}^{T}=b_{1}^{T})=0$ is zero, it follows from (\ref{O64}) that for $\tau\geq t,$
\begin{align} \label{O65}
	&\quad\ (1-\lambda_{\tau-t})P(X_{1}^{t}=a_{1}^{t})P(X_{\tau+1}^{\tau+T}=b_{1}^{T}) \nonumber \\
	&\leq P(X_{1}^{t}=a_{1}^{t},\ X_{\tau+1}^{\tau+T}=b_{1}^{T}) \nonumber \\
	&\leq(1+\lambda_{\tau-t})P(X_{1}^{t}=a_{1}^{t})P(X_{\tau+1}^{\tau+T}=b_{1}^{T})
\end{align}
Denote
\begin{align}\label{O66}
C \triangleq \sum_{\tau=0}^{t-1}P(X_{1}^{t}=a_{1}^{t},\ X_{\tau+1}^{\tau+T}=b_{1}^{T})
\end{align}
It follows from (\ref{O65}) by taking a sum over $\tau$ that and by noting that since the process is stationary, $P(X_{\tau+1}^{\tau+T}=b_{1}^{T})=P(X_{1}^{T}=b_{1}^{T})$ and substituting (\ref{O66}) in (\ref{O65})
\begin{align} \label{O67}
	&\quad\ C+\left(N-t-\displaystyle \sum_{\tau=t}^{N-1}\lambda_{\tau-t}\right)P(X_{1}^{t}=a_{1}^{t})P(X_{1}^{T}=b_{1}^{T}) \nonumber \\
	&\leq \sum_{\tau=0}^{N-1}P(X_{1}^{t}=a_{1}^{t},\ X_{\tau+1}^{\tau+T}=b_{1}^{T}) \nonumber \\
	&\leq C+\left(N-t+\sum_{\tau=t}^{N-1}\lambda_{\tau-t}\right)P(X_{1}^{t}=a_{1}^{t})P(X_{1}^{T}=b_{1}^{T})
\end{align}
After noting that $C$ and $t$ are constants, that $\lambda_{\tau}\rightarrow 0$ as $\tau\rightarrow\infty$, after dividing by $N$ and taking limits as $ N\rightarrow\infty$ in (\ref{O67}) , it follows that
\begin{align}\label{O68}
\lim_{N\rightarrow\infty}\sum_{\tau=0}^{N-1}P(X_{1}^{t}=a_{1}^{t},\ X_{\tau+1}^{\tau+T}=b_{1}^{T})=P(X_{1}^{t}=a_{1}^{t})P(X_{1}^{T}=b_{1}^{T})
\end{align}
thus proving (\ref{O62}) , and thus, proving that the process $X$ is ergodic if it is stationary, $\psi$-mixing.
\end{proof}

Proof of Lemma \ref{L3} :
\begin{proof}
Consider the two-sided extension $V=(V_{t},\ t\in \mathbb{Z})$ of $X$, defined on a probability space $(\Omega'', \Sigma'', P'')$. That is,
\begin{align} \label{O69}
P''(V_{t+1}=j|V_{t}=i)=p_{ij},\quad -\infty<t<\infty
\end{align}
where $p_{ij}$ denotes the probability
\begin{align}\label{O70}
P(X_{t+1}=j|X_{t}=i) ,\quad 1\leq t<\infty
\end{align}
which is independent of $t$ since $X$ is Markoff. Such an extension is possible, see for example \cite{DED}. Denote by $\mathbb{X}^{\mathbb{Z}}$, the set of doubly-infinite sequences taking values in  $\mathbb{X}$. The Borel-sigma field on $\mathbb{X}^{\mathbb{Z}}$ is the standard construction, see Pages 1-5 of \cite{Shields}. Note that $V$ is finite-state, stationary, irreducible, aperiodic.

Denote the Borel-sigma field on $\mathbb{X}^{\mathbb{Z}}$ by $\mathcal{H}_{-\infty}^{\infty}$ and as was the case when defining $\mathcal{F}_{a}^{b}$, denote the Borel sigma-field on $\mathbb{X}_{a}^{b}$ by $\mathcal{H}_{a}^{b}, -\infty\leq a\leq b\leq\infty.$

For the process $V$, consider the standard definition of $\psi$-mixing as stated in \cite{Bradley}, and thus, define
\begin{multline}\label{O71}
\psi_{V}(\tau) \triangleq
\sup_{t\in \mathbb{Z}}
\sup_{\mathbb{K}\in \mathcal{H}_{-\infty}^{t},\mathbb{L}\in \mathcal{H}_{t+\tau+1}^{\infty},P''(V_{-\infty}^{t}\in \mathbb{K})>0,P''(V_{t+\tau+1}^{\infty}\in \mathbb{L})>0}\\
\left |\frac{P''(V_{-\infty}^{t}\in \mathbb{K},V_{t+\tau+1}^{\infty}\in \mathbb{L})}{P''(V_{-\infty}^{t}\in \mathbb{K})P''(Z_{t+\tau+1}^{\infty}\in \mathbb{L})}-1 \right |
\end{multline}
The process $V$ is said to be $\psi$-mixing if $\psi_{V}(\tau)\rightarrow 0$ as $\tau\rightarrow\infty$. Since $V$ is stationary, irreducible, aperiodic, finite-state Markoff chain, by Theorem 3.1 of \cite{Bradley}, $V$ is $\psi$-mixing.

Let $\mathbb{A}\in \mathcal{F}_{1}^{t}$. Consider the set $\mathbb{A}''$ defined as follows:
\begin{align}\label{O72}
 \mathbb{A}''= \{(\ldots, a_{-n}, \ldots a_{-1}, a_{0}, {a}_{1}, \ldots, a_{t})|(a_{1}, {a}_{2}, \ldots, a_{t})\in \mathbb{A}\}
\end{align}
Then, since $X$ is stationary and $V$ is the double-sided extension of $X,$
\begin{align} \label{O73}
 P''(V_{-\infty}^{t} \in \mathbb{A}'') =P(X_{1}^{t}\in \mathbb{A})
\end{align}
and by use of the Markoff property, and again, noting that $V$ is the double- sided extension of $X$, it follows that
\begin{align} \label{O74}
P''(V_{-\infty}^{t}\in \mathbb{A}'', V_{t+\tau+1}^{\infty}\in \mathbb{B})=P(X_{1}^{t}\in \mathbb A,\ X_{t+\tau+1}^{\infty}\in \mathbb{B})
\end{align}
%28
By use of (\ref{O73}) and (\ref{O74}), it follows that
\begin{align}\label{O75}
\left | \frac{P''(V_{-\infty}^{t}\in \mathbb{A}'',V_{t+\tau+1}^{\infty}\in \mathbb{B})}{P''(V_{-\infty}^{t}\in \mathbb{A}'')P''(V_{t+\tau+1}^{\infty}\in \mathbb{B})}-1 \right |= \left |\frac{P(X_{1}^{t}\in \mathbb{A},X_{t+\tau+1}^{\infty}\in \mathbb{B})}{P(X_{1}^{t}\in \mathbb{A})P(X_{t+\tau+1}^{\infty}\in \mathbb{B})}-1 \right |
\end{align}
where $\mathbb{A}\in \mathcal{F}_{1}^{t}, \mathbb{B}\in \mathcal{F}_{t+\tau+1}^{\infty}, P(X_{1}^{t}\in \mathbb{A})>0, P(X_{t+\tau+1}^{\infty}\in \mathbb{B})>0.$

Thus,
\begin{align} \label{O76}
	&\quad\sup_{\mathbb{A}\in \mathcal{F}_{1}^{t},\mathbb{B}\in \mathcal{F}_{t+\tau+1}^{\infty},P(X_{1}^{t}\in \mathbb{A})>0,P(X_{t+\tau+1}^\infty \in \mathbb{B})>0} 
\left |\frac{P''(V_{-\infty}^{t}\in \mathbb{A}'',V_{t+\tau+1}^{\infty}\in \mathbb{B})}{P''(V_{-\infty}^{t}\in \mathbb{A}')P''(V_{t+\tau+1}^{\infty}\in \mathbb{B})}-1 \right | \\
	&= 
\sup_{\mathbb{A}\in \mathcal{F}_{1}^{t},\mathbb{B}\in \mathcal{F}_{t+\tau+1}^{\infty},P(X_{1}^{t}\in \mathbb{A})>0,P(x_{t+\tau+1}^{\infty}\in \mathbb{B})>0}
\left |\frac{P(X_{1}^{t}\in \mathbb{A},X_{t+\tau+1}^{\infty}\in \mathbb{B})}{P(X_{1}^{t}\in \mathbb{A})P(X_{t+\tau+1}^{\infty}\in \mathbb{B})}-1 \right |
\nonumber
\end{align}
Thus,
\begin{align}\label{O77}
	&\sup_{ \mathbb{K}\in \mathcal{G}_{-\infty}^t, \mathbb{L}\in \mathcal{G}_{t+\tau+1}^{\infty},P''(V_{1}^{t}\in \mathbb{K})>0,P''(V_{t+\tau+1}^{\infty}\in \mathbb{L})>0}
\left |\frac{P''(V_{-\infty}^{t}\in \mathbb{K},V_{t+\tau+1}^{\infty}\in \mathbb{L})}{P''(V_{-\infty}^{t}\in \mathbb{K})P''(V_{t+\tau+1}^{\infty}\in \mathbb{L})}-1 \right |  \\
	\notag&\geq
\sup_{\mathbb{A}\in \mathcal{F}_{1}^{t},\mathbb{B}\in \mathcal{F}_{t+\tau+1}^{\infty},P(X_{1}^{t}\in \mathbb{A})>0,P(X_{t+\tau+1}^{\infty}\mathbb{B})>0}
\left |\frac{P(X_{1}^{t}\in \mathbb{A},X_{t+\tau+1}^{\infty}\in \mathbb{B})}{P(X_{1}^{t}\in \mathbb{A})P(X_{t+\tau+1}^{\infty}\in \mathbb{B})}-1 \right |
\end{align}
This is because there are sets $\mathbb{K}\in \mathcal{G}_{-\infty}^{t}$ and $\mathbb{L}\in \mathcal{G}_{t+\tau+1}^{\infty}$ which are not of the form $\mathbb{A}''$ and $\mathbb{B}''$ respectively.

Denote the function $\psi$, defined in (\ref{O18}) for the process $X$ by $\psi_{X}$. It follows from (\ref{O77}) that $\psi_{Z}(\tau)\geq\psi_{X}(\tau)$ . Since $Z$ is $\psi$-mixing as stated above, by definition, $\psi_{Z}(\tau)\rightarrow 0$ as $\tau\rightarrow\infty$. Thus, $\psi_{X}(\tau)\rightarrow 0$ as $\tau\rightarrow\infty$, and thus, $X$ is $\psi$-mixing.
\end{proof}

Proof of Lemma \ref{L4}:
\begin{proof}
Stationary of $Z$ follows directly from the definition of stationarity.

Denote the $\psi$ function for $X$ and $Z$ by $\psi_{X}$ and $\psi_{Z}$ respectively. Note that
%29
the $\psi$ function for the process $Z$ can be written as follows:
\begin{multline}\label{O78}
\psi_{Z}(\tau) \triangleq 
\sup_{t \in \mathbb N}
\sup_{\mathbb{A}\in \mathcal{F}_{1}^{tL},\mathbb{B}\in \mathcal{F}_{tL+\tau L+1}^{\infty},P(X_{1}^{tL}\in \mathbb{A})>0,P(X_{tL+\tau L+1}^{\infty}\in \mathbb{B})>0}\\
\left |
\frac
{P(X_{1}^{tL}\in \mathbb{A},X_{tL+\tau L+1}^{\infty}\in \mathbb{B})}
{P(X_{1}^{tL}\in \mathbb{A})P(X_{tL+\tau L+1}^{\infty}\in \mathbb{B})}
-1
	\right |
\end{multline}
Note that when calculating the $\psi$ function for $Z$, the supremum is taken over a lesser number of sets than when calculating the $\psi$ function for $X$. It follows that $\psi_{Z}(\tau)\leq\psi_{X}(\tau)$. Since $X$ is $\psi$-mixing, $\psi_{X}(\tau)\rightarrow 0$ as $\tau\rightarrow\infty$. It follows that $\psi_{Z}(\tau)\rightarrow 0$ as $\tau\rightarrow\infty$. Thus, $Z$ is $\psi$-mixing.
\end{proof}
Proof of Lemma \ref{L5}:
\begin{proof}
 Note that $Z$ is stationary by Lemma \ref{L4}. Thus, $Z$ is a stationary, irreducible, aperiodic, finite-state Markoff chain, evolving on a finite set, and by Lemma \ref{L3} , $\psi$-mixing.

Since the set $\mathbb{Z}$ is finite, the Borel sigma field on $\mathbb{Z}^{\infty}$ can be constructed analogously to that on $\mathbb{X}^{\infty}$; see Page 1-2 of \cite{Shields}. Denote this Borel sigma field by $\mathcal{G}^{\infty}$. Define the Borel sigma fields $\mathcal{G}_{a}^{b}$, analogously as was done for $\mathcal{F}_{1}^{\infty}.$ Denote the underlying probability space by $(\Omega', \Sigma', P')$

An element of $\mathbb{Z}^{\infty}$ is denoted by $(z_{1}, z_2, \ldots)$ where $z_{i}\in \mathbb{Z}=\mathbb{X}^{L}$. The $j^{th}$ component of $z_{i}$ will be denoted by $z_{i}(j)$ .

Define
\begin{multline} \label{O79}
\psi_{X}(\tau) \triangleq
\sup_{t \in \mathbb N}
\sup_{\mathbb{A}\in \mathcal{F}_{1}^{t},\mathbb{B}\in \mathcal{F}_{t+\tau+1}^{\infty},P(X_{1}^{t}\in \mathbb{A})>0,P(X_{t+\tau+1}^{\infty}\in \mathbb{B})>0}\\
\left |
\frac{P(X_{1}^{t}\in \mathbb{A},X_{t+\tau+1}^{\infty}\in \mathbb{B})}
        {P(X_{1}^{t}\in \mathbb{A})P(X_{t+\tau+1}^{\infty}\in \mathbb{B})}
-1
	\right |
\end{multline}
and
\begin{multline}\label{O80}
\psi_{Z}(\tau)= \triangleq
\sup_{t\in \mathbb{N}}
\sup_{\mathbb{A}'\in \mathcal{G}_{1}^{t},\mathbb{B}'\in \mathcal{G}_{t+\tau+1}^{\infty},P'(Z_{1}^{t}\in \mathbb{A}')>0,P'(Z_{t+\tau+1}^{\infty}\in \mathbb{B}')>0}\\
\left |
\frac{P'(Z_{1}^{t}\in \mathbb{A}',Z_{t+\tau+1}^{\infty}\in \mathbb{B}')}
        {P'(Z_{1}^{t}\in \mathbb{A}')P'(Z_{t+\tau+1}^{\infty}\in \mathbb{B}')}
-1
	\right |
\end{multline}

By definition, the processes $X$ and $Z$ are $\psi$-mixing if $\psi_{X}(\tau)$ and $\psi_{Z}(\tau)$ tend to zero as $\tau\rightarrow\infty$, respectively.

%30
For $\mathbb{A}\in \mathcal{F}_{1}^{t}, \mathbb{B}\in \mathcal{F}_{t+\tau+1}^{\infty}, P(X_{1}^{t}\in \mathbb{A})>0, P(X_{t+\tau+1}^{\infty}\in \mathbb{B})>0$, define,
\begin{align}\label{O81}
\kappa_{X}(t,\ \tau,\ \mathbb{A},\ \mathbb{B})
\triangleq
\left |
\frac
{P(X_{1}^{t}\in \mathbb{A},X_{t+\tau+1}^{\infty}\in \mathbb{B})}
{P(X_{1}^{t}\in \mathbb{A})P(X_{t+\tau+1}^{\infty}\in \mathbb{B})}
-1
\right |
\end{align}

Define
\begin{align}\label{O82}
k_{1} & \triangleq \left \lceil\frac{t}{L}\right \rceil \\
 k_{2}& \triangleq \left \lfloor\frac{t+\tau+1}{L} \right \rfloor \nonumber
\end{align}

Assume that $\tau\geq 4L$. It follows that $k_{1}\leq k_{2}$ (a weaker assumption is possible, but this suffices).

Given A and $\mathbb{B}$, define $\mathbb{A}'$ and $\mathbb{B}'$ by
\begin{align} \label{O83}
\mathbb{A}' & \triangleq \{(a_{1}, {a}_2, \ldots, a_{k_{1}L})|(a_{1}, a_2, \ldots, a_t) \in \mathbb A \} \\
\mathbb{B}' & \triangleq \{(b_{k_{2}L+1}, b_{k_{2}L+2}, \ldots)|(b_{t+\tau+1}, b_{t+\tau+2}, \ldots)\in \mathbb{B}\} \nonumber 
\end{align}

Think, now of $(a_1, \ldots, a_{k_1L})$ as $a'=(a'_1, \ldots, a'_{k_1})$, a $k_1$ length sequence, where $a'_i \in \mathbb Z$. This can be done by defining $a_{i}'=a_{(i-1)L+1}^{iL}$. Analogously, think of $(b_{k_{2}L+1}', b_{k_{2}L+2}', \ldots)$ as $(b_{k_{2}+1}',b_{k_{2}+2}', \ldots)$ where $b_{k_{i}}'$ is defined analogously to how $a_{i}'$ was defined. Think of $\mathbb{A}'$ and $\mathbb{B}'$, now, as sequences of elements in $\mathbb{Z}$ in the obvious way.

Define, for $\mathbb{J}\in \mathcal{G}_{1}^{q}, \mathbb{U}\in \mathcal{G}_{q+q'+1}^{\infty},$
\begin{align} \label{O84}
\kappa_{Z}(q,\ q',\ \mathbb{J},\ \mathbb{U})\triangleq 
\left |
\frac
{P'(Z_{1}^{t}\in\mathbb J,Z_{t+\tau+1}^{\infty}\in \mathbb{U})}
{P'(Z_{1}^{t}\in\mathbb J)P'(Z_{t+\tau+1}^{\infty}\in \mathbb{U})}
-1
\right |
\end{align}
Then, it follows that for $\tau\geq 4L,$
\begin{align} \label{O85}
\kappa_{X}(t, \tau, \mathbb{A}, \mathbb{B})=\kappa_{Z}(k_{1}, k_{2}-k_{1}, \mathbb{A}', \mathbb{B}')
\end{align}
%31
Denote
\begin{multline}\label {O86}
\mu_X(t,\ \tau) =
\sup_{\mathbb{A}\in \mathcal{F}_{1}^{t},\mathbb{B}\in \mathcal{F}_{t+\tau+1}^{\infty},P(X_{1}^{t}\in \mathbb{A})>0,P(X_{t+\tau+1}^{\infty}\in \mathbb{B})>0}\\
\left |
\frac
{P(X_{1}^{t}\in \mathbb{A},X_{t+\tau+1}^{\infty}\in \mathbb{B})}
{P(X_{1}^{t}\in \mathbb{A})P(X_{t+\tau+1}^{\infty}\in \mathbb{B})}
-1
	\right |
\end{multline}
\begin{multline*}
 \mu_{Z}(q,\ q')=
\sup_{\mathbb{J}\in \mathcal{G}_{1}^{q},\mathbb{U}\in \mathcal{G}_{q+q'+1}^{\infty},P'(Z_{1}^{q}\in \mathbb{J})>0,P'(Z_{t+\tau+1}^{\infty}\in \mathbb{U})>0}\\
\left |
\frac
{P'(Z_{1}^{t}\in\mathbb J, Z_{t+\tau+1}^{\infty}\in \mathbb{U})}
{P'(Z_{1}^{t}\in\mathbb J)P'(Z_{t+\tau+1}^{\infty}\in \mathbb{U})}
-1
	\right |
\end{multline*}
It follows from (\ref{O85}) by taking supremum over sets $\mathbb{A}\in \mathcal{F}_{1}^{t}$ and $\mathbb{B}\in \mathcal{F}_{t+\tau+1}^{\infty}$ and then, noting that there are sets $\mathbb{J}\in \mathcal{G}_{1}^{k_{1}}$ and $\mathbb{U}\in \mathcal{G}_{k_{2}+1}^{\infty}$ which are not of the form $\mathbb{A}'$ and $\mathbb{B}'$, that
\begin{align} \label{O87}
 \mu_{X}(t, \tau) \leq
\mu_{Z}(k_{1}, k_{2} - k_{1})=
\mu_{Z}\left (\left \lceil\frac{t}{L} \right \rceil, \left \lfloor\frac{t+\tau+1}{L} \right \rfloor- \left \lceil\frac{t}{L} \right \rceil \right )
\end{align}
Thus,
\begin{align} \label{O88}
\psi_{X}(\tau)\leq
\sup_{t\in \mathbb{N}}
\mu_{Z}\left (\left \lceil\frac{t}{L}\right \rceil, \left \lfloor\frac{t+\tau+1}{L}\right \rfloor - \left \lceil\frac{t}{L}\right \rceil \right )
\end{align}
The right hand side in the above equation $\rightarrow 0$ as $\tau\rightarrow\infty$ since $Z$ is $\psi-$ mixing. Thus, $\psi_{X}(\tau)\rightarrow 0$ as $\tau\rightarrow\infty$, and thus, $X$ is $\psi$-mixing.
\end{proof}

\bibliography{2-Agarwal}

\address{Navi Mumbai, 410210, Maharashtra, India\\
\email{magar@alum.mit.edu}}

\address{Laboratory for information and decision systems \\
Department of Electrical Engineering and Computer Science \\
Massachusetts Institute of Technology\\
Cambridge, MA 02139-4307, USA\\
\email{mitter@mit.edu}}

\address{Department of Electrical Engineering and Computer Sciences\\
University of California, Berkeley\\
 Berkeley, CA 94720-1770, USA\\
\email{sahai@eecs.berkeley.edu}}

\end{document}